\documentclass[journal,twoside,web]{ieeecolor}
\usepackage{generic}
\usepackage{cite}
\usepackage{amsmath,amssymb,amsfonts}
\usepackage{algpseudocode,algorithm,algorithmicx}
\usepackage{graphicx}
\usepackage{textcomp}

\newtheorem{secthm}{Theorem}[section]
\newtheorem{seccor}[secthm]{Corollary}

\newtheorem{secprop}[secthm]{Proposition}
\newtheorem{secprob}[secthm]{Problem}
\newtheorem{secdefn}[secthm]{Definition}
\newtheorem{secrem}[secthm]{Remark}

\newcommand{\cA} { {\mathcal A}}

\newcommand{\cK} { {\mathcal K}}

\newcommand{\bR} { {\mathbb R}}

\newcommand{\offd}[1] {\mbox{\rm off-diag}(#1)}
\newcommand{\1}{\mbox{1}\hspace{-0.25em}\mbox{l}}

\def\red{\hfill $\lhd$}

\def\BibTeX{{\rm B\kern-.05em{\sc i\kern-.025em b}\kern-.08em
    T\kern-.1667em\lower.7ex\hbox{E}\kern-.125emX}}
\begin{document}
\title{$\mathcal{K}$-monotonicity and feedback synthesis \\ for incrementally stable networks}
\author{Yu Kawano, Fulvio Forni
\thanks{This work was supported in part by JSPS KAKENHI Grant Number JP21K14185.}
\thanks{Y. Kawano is with the Graduate School of Advanced Science and Engineering, Hiroshima University, 
Higashi-Hiroshima 739-8527, Japan (e-mail: ykawano@hiroshima-u.ac.jp). }
\thanks{F. Forni is with the Department of Engineering, University of Cambridge,
Cambridge, CB2 1PZ, UK (e-mail: f.forni@eng.cam.ac.uk).}}

\maketitle

\begin{abstract}
We discuss the role of monotonicity in enabling numerically tractable modular control design for networked nonlinear systems. 
We first show that the variational systems of monotone systems can be embedded into positive systems.
Utilizing this embedding, we show how to solve a network stabilization problem by enforcing
monotonicity and exponential dissipativity of the network sub-components. 
Such modular approach leads to a design algorithm based on a sequence
of linear programming problems. 
\end{abstract}


\section{Introduction}
A linear system is said to be positive if it maps positive state vectors into positive state vectors,
that is, if the positive orthant $\bR^n_+$ is a forward invariant set for the system dynamics.  
The notion of positivity is often extended to linear systems that admit a generic forward invariant proper cone $\mathcal{K}$
\cite{Smith:08,HS:03,Farina:00}. Perron-Frobenius theory shows that a positive linear system has a dominant (slowest) mode 
constraining its asymptotic behavior to a one-dimensional ray~\cite{Bushell1973,Bushell1973a,HS:03}.
This particular feature allows to study stability and dissipativity 
of these systems using linear forms instead of quadratic ones, 
opening the door to scalable algorithms \cite{Rantzer:15,HCH:10}.

In the nonlinear setting the corresponding class of systems 
are called monotone systems \cite{Smith:08,HS:03,AS:03,FS:16}.
These systems  are characterized by the property that their trajectories preserve a partial order on the
system state space. Nonlinear monotone systems show a strict connection with positivity:
their variational flow guarantees forward invariance of some proper cone $\mathcal{K}$ \cite{FS:16}.
Like for linear systems, monotonicity makes nonlinear analysis scalable \cite{DIR:15,FS:16,WMB:17,FBJ:18,Coogan:19,KBC:20,KB:22}.
However, certifying the monotonicity of a system is generally difficult. From a variational
perspective, the problem corresponds to the determination of a proper cone $\mathcal{K}$ 
that is forward invariant for the variational dynamics \cite{FS:16}. 

The objective of this paper is to take advantage of monotonicity to 
design feedback controllers that  enforce incremental stability and differential 
dissipativity of systems. This leads to numerical methods for modular control of networked nonlinear systems.
Our approach consists of two steps. First, we certify that each network sub-component 
is monotone by showing that its variational flow  
can be represented, via embedding, by the flow of a positive linear time-varying system.
Then, using the embedding, we take advantage of scalable conditions
for stability and dissipativity for monotone systems to characterize network properties.
This leads to conditions for exponential stability of a network of monotone systems 
based on the interconnection structure and on the supply rate of each sub-component 
of the network. 

The problem of deriving the embedding to certify the monotonicity of network
sub-components can be solved by finding a forward invariant cone for their variational dynamics. 
In this paper we take inspiration from the early results of \cite{KF:20}
to develop a numerical algorithm for modular control. The algorithm
co-designs cone and feedback controller to guarantee 
closed-loop monotonicity and exponentially dissipativity.
The algorithm consists of an iteration loop of linear programming problems,
enabling numerically tractable modular design of nonlinear networks.

The novelty of the paper lies in the following results:
(i) conditions to certify monotonicity with respect to partial orders induced by generic cones $\cK$; 
(ii) conditions for incremental exponential stability and exponential differential dissipativity of monotone systems;
and (iii) conditions for incremental stability of networks of monotone systems.
These results can be used both for system analysis and
for control synthesis. In fact, when the cone $\cK$ is known, the conditions of the paper 
reduce to linear programs, which are easy to solve. When the cone $\cK$ is not known, the two algorithms 
of the paper provide procedures for co-derivation of cone and feedback controller,
with the goal of enforcing monotonicity together with incremental stability (closed systems) or differential dissipativity (open systems).

The results of this paper make contact with network stability analysis for positive linear systems \cite{HCH:10, Rantzer:15}.
Specifically, \cite{HCH:10} uses exponential dissipativity to certify the stability of the
positive feedback interconnection of two positive systems, and 
\cite{Rantzer:15} looks at distributed positivity-preserving control design for stability. 
The centralized case has been studied in \cite{RT:07}, where linear programming is used for control design. 
Our paper can be viewed as a direct extension of these earlier results to generic cones $\cK$ (in contrast to the positive orthant),
both for linear and nonlinear systems. It demonstrates how the scalable conditions in \cite{HCH:10, Rantzer:15, RT:07} 
go beyond the positive orthant and linearity.
Our approach is based on embedding a variational system into a positive variational system, which 
corresponds to a (variational) positive realization problem. In the linear case, the forward invariance of some proper 
cone $\mathcal{K}$ is known to be equivalent to the existence of a positive realization \cite{Farina:95, Hof:97, ADF:96, BF:02, BF:04}.
However, in literature, the problem of finding a positive realization and the stabilization problem of a positive system are typically studied in autonomy.
The novelty of this paper is to combine these objectives to arrive at scalable stability and dissipativity conditions for both linear and nonlinear systems.

The remainder of this paper is organized as follows.
After a brief discussion of notation and mathematical preliminaries,
Section~\ref{sec:ME} provides a motivating example. The goal is to justify
our effort towards scalable conditions for analysis and synthesis of positive (linear) and monotone (nonlinear) systems.
Section~\ref{sec:ML} introduces the class of positive linear systems
and derive technical conditions for their stability and dissipativity. The latter opens
the way to the characterization of stable networks. Most of the results of Section~\ref{sec:ML}
are generalized to nonlinear monotone systems in Section~\ref{sec:MN}.
In particular, Section~\ref{sec:MN} provides conditions for monotonicity, incremental exponential stability (or contraction) and differential dissipativity.
Finally, Section~\ref{sec:AC} takes advantage of these results to provide numerical algorithms for system analysis and control design.
All proofs are in appendix.
\vspace{2mm}

\noindent
\textbf{Notation and mathematical preliminaries:}

The vector whose all elements are $1$ is denoted by $\1$.
The vector $1$-norm is denoted by $| \cdot |_1$.
The identity matrix is denoted by $I$.
For matrices $A,B \in \bR^{n \times m}$, 
we write $A \le B$ ($A < B$) if and only if 
$A_{i,j} \le B_{i,j}$ ($A_{i,j} < B_{i,j}$) for all $i=1,\dots,n$ and $j=1,\dots,m$, where 
$A_{i,j}$ denotes the $(i,j)$-th component of $A$. 
Similarly, for $c \in \bR$ and $A \in \bR^{n \times m}$, $c \le A$ ($c < A$) means
$c \le A_{i,j}$ ($c < A_{i,j}$) for all $i=1,\dots,n$ and $j=1,\dots,m$.
For a square matrix~$A \in \bR^{n \times n}$, 
$\offd{A} \ge 0$ means its off-diagonal elements are all non-negative.
Also, $|\offd{A}|$ denotes the matrix obtained by taking the absolute value of each off-diagonal elements of $A$.

A subset~$\cK \subset \bR^n$ is called a \emph{simple cone} if it is
(i) a cone, i.e. $c \cK \subset \cK$ for any scalar~$c \ge 0$,  
(ii) closed; 
(iii) convex, i.e. $\cK + \cK \subset \cK$; and
(iv) pointed, i.e. $-\cK\cap \cK =\{0\}$.
Moreover, a simple cone is called \emph{proper} if it is
(v) solid, i.e. ${\rm int}(\cK ) \neq \emptyset$. 

A cone $\cK \subset \bR^n$ is said to be a \emph{polyhedral cone} if 
there exists a finite set of vectors $K_i \in \bR^n \setminus \{0\}$, for $i =1,\dots,p_{\cK}$ such that
\begin{align*}
&\cK = \{x \in \bR^n:  K x \ge 0\} \mbox{, where }
K := \begin{bmatrix} K_1 & \cdots &K_{p_{\cK}} \end{bmatrix}^\top.
\end{align*}
A polyhedral cone is closed and convex.
A polyhedral cone $\cK$ is simple if and only if ${\rm rank} \; K = n$.
For example, the canonical basis $e_i \in \bR^n$, $i\in \{1,\dots,n\}$, corresponds to the positive orthant $\bR_+^n$.
For more details about cones, see \cite{BP:94}.

\section{Motivating Example: Mass-Spring system}\label{sec:ME}
Consider the following mass-spring system: 
\begin{align}\label{sys:mech}
\left\{\begin{array}{r@{}l}
\begin{bmatrix}
\dot x_1 \\ 
\dot x_2
\end{bmatrix}
&{}=
\begin{bmatrix}
0 & 1 \\
-k & 0
\end{bmatrix}
\begin{bmatrix}
x_1 \\ 
x_2
\end{bmatrix}
+
\begin{bmatrix}
0 \\ 
1
\end{bmatrix} \bar u, 
\quad 
k>0 \\
y &{}= x_1.
\end{array}\right.
\end{align}
We investigate a feedback controller:
\begin{align*}
\bar u = F_1 x_1 + F_2  x_2 + u
\end{align*}
to make the closed-loop system  positive with respect to the positive orthant $\bR^n_+$ and exponentially stable. 
This form of positivity requires
\begin{align}\label{pos:mech}
F_1 - k \ge 0,
\end{align}
since, by the Kamke condition \cite{Smith:08,AS:03}, the off-diagonal elements of the state matrix
must be non-negative. 
For exponential stability, the eigenvalues of the closed-loop system
\begin{align*}
\lambda = \frac{F_2 \pm \sqrt{F_2^2 + 4 (F_1 - k)}}{2}
\end{align*} 
have the negative real parts if and only if
\begin{align}\label{ES:mech}
\left\{\begin{array}{l}
F_2 < 0\\
F_1 - k < 0.
\end{array}\right.
\end{align}
Thus, positivity with respect to the positive orthant and exponential stability \eqref{ES:mech} are 
incompatible for the controlled mass-spring system.

Apparently these conditions suggest that exponential stability can be enforced at the cost of waiving positivity 
with respect to the positive orthant and its scalable design features \cite{Rantzer:15,HCH:10}.
However, as will be shown later, scalability can be retained by relying on positivity with respect to the polyhedral cone~$\cK \subseteq \bR^2$
 defined by
\begin{align}\label{cone:mech}
K
= 
\begin{bmatrix}
K_1 & K_2
\end{bmatrix}^\top,
\quad
K_1:= \begin{bmatrix} 1 \\ 0 \end{bmatrix}, \;
K_2 :=  \begin{bmatrix} 2 \\ 1 \end{bmatrix}.
\end{align}
According to Theorem~\ref{thm:LM} below, the closed-loop system is positive with respect to $\cK$, 
(i.e. $\cK$ is a forward invariant set of the closed-loop system) \cite{Smith:08, AS:03, HS:06}
if and only if
\begin{align}\label{CI:mech}
F_2 \le (- 4 + F_1 - k)/2.
\end{align} 
Notably, \eqref{ES:mech} and \eqref{CI:mech} characterize a \emph{feasible} set of linear constraints in $F_1$ and $F_2$, 
ensuring positivity with respect to $\cK$ and exponential stability of the controlled mass-spring system.

\begin{secrem}
If a pair $(A, B)$ is stabilizable, there exists a feedback gain $F$ rendering the closed-loop system $\dot x = (A + BF) x$ both exponentially stable and $\cK$-positive with respect to a proper cone. From stabilizability, we can shift all unstable eigenvalues of $A$ to arbitrary values by selecting suitable $F$. In particular, one can make $A+ B F$ has the simple dominant eigenvalue $0> \lambda_1 \in \bR$, i.e., $\lambda_1 > {\rm Re} (\lambda_i)$ for the other eigenvalues $\lambda_i$ of $A+ B F$. This implies that all eigenvalues of $A+ B F$ are in the left-half plane and the exponential stability. Next, $e^{\lambda_1 t} > |e^{\lambda_i t}|$ for all $\lambda_i$, $i \neq 1$. According to \cite[Theorem 3.5]{BP:94}, this implies $e^{(A+ B F)t} \cK \subset \cK$, i.e., $\cK$-positivity with respect to a proper cone $\cK$. Moreover, $\cK$ is time-invariant  from its construction in \cite[Theorem 3.5]{BP:94}, since the (generalized) eigenvectors of $e^{(A+ B F)t}$ and $A+ B F$ are the same and time-invariant. In this paper, to develop numerically tractable algorithms, we require a proper cone $\cK$ to be polyhedral. A proper cone can always be approximated by a polyhedral cone by increasing $p_\cK$.~\red
\end{secrem}

An important observation is that 
the change of coordinates $z = K x$
yields a closed-loop system that is positive with respect to $\bR_+^2$:
\begin{align*}
\left\{\begin{array}{r@{}l}
\dot z &{}=
\begin{bmatrix}
-2 & 1 \\
 - k - 4 + F_1 - 2 F_2 & 2 + F_2
\end{bmatrix}
z
+
\begin{bmatrix}
0 \\ 
1
\end{bmatrix} u \\[2.5mm]
y &{}= \begin{bmatrix}
1 & 0
\end{bmatrix}
z.
\end{array}\right.
\end{align*}
This shows how scalable approaches for positive systems with respect to the positive orthant \cite{Rantzer:15,HCH:10} 
are directly applicable to general positive systems through the mapping $z = K x$. 
Fundamentally, this example clarifies how the unfeasibility of \eqref{pos:mech} and \eqref{ES:mech} 
is related to the coordinate representation of the system and does not describe 
an intrinsic feature of the system.

In general, $K$ is not necessarily square but always satisfies ${\rm rank} \; K = n$ for $x \in \bR^n$ if $\cK$ is a proper (or more generally, simple) polyhedral cone.
This implies that $z = K x$ is an embedding yielding 
a positive realization in the $z$-coordinates.
This has two important consequences: 
1) scalable conditions for positive linear systems with respect to the positive orthant are directly applicable to general 
positive systems by working in the $z$-coordinates, via embedding;
2) stability in the $z$-coordinates implies stability in the $x$-coordinates.

Encouraged by these observations, in what follows, we take advantage of the \emph{embedding trick} to 
develop scalable analysis and control design methods for positive linear systems and for monotone nonlinear systems.
We will start by assuming that the forward invariant cone $\mathcal{K}$ is given.
Later on, we will propose an algorithm to co-design both cone and controller 
to enforce monotonicity and exponential stability simultaneously.

\section{Positive Linear Systems}\label{sec:ML}

\subsection{Positivity and Stability}
In this section we recall the definition of positivity for linear systems
and we show how scalable methods for positive linear systems with respect to the positive orthant  \cite{Rantzer:15,HCH:10} 
can be adapted to any positive linear system.

\begin{secdefn}
\label{def:positivity}
The open linear system $\Sigma$
\begin{align}
\Sigma: 
\left\{\begin{array}{r@{}l}
\dot x &{}= A x + B u\\
y &{}= C x
\end{array}\right.
\quad
(x, u, y) \in \bR^n \times \bR^m \times \bR^p
\end{align}
is said to be \emph{positive} with respect to 
a triplet of polyhedral cones $(\cK, \bR_+^m, \bR_+^p)$ if,
for all $(x(\cdot),u(\cdot),y(\cdot)) \in \Sigma$,
\begin{align}
\begin{array}{c}
x (0) \in \cK, \; u(t) \in \bR_+^m,			\ \forall t \!\ge \! 0. \\
\implies \\
 x (t) \in \cK, \; y(t) \in \bR_+^p,	\ \forall t \!\ge \! 0
\end{array}
\end{align}

For a closed system $\dot x = Ax$, positivity simply means that
$x (0) \in \cK$ implies $x (t) \in \cK$ for any $t \ge 0$.
\red
\end{secdefn}

In what follows we will adopt the simpler terminology of $\mathcal{K}$\emph{-positive linear systems}
for the systems in Definition \ref{def:positivity}, leaving the input and output cones implicit.

According to \cite{Rantzer:15,HCH:10}, a necessary and sufficient conditions for a linear system to be $\bR_+^n$-positive is
\begin{align}\label{cond:PR}
\offd{A} \ge 0,
\quad
B \ge 0,
\quad
C \ge 0.
\end{align}
The following theorem provides equivalent conditions for general $\mathcal{K}$-\emph{positive} linear systems.

\begin{secthm}[positivity]\label{thm:LM}
An open linear system $\Sigma$ is $\cK$-positive 
if and only if there exist $P  \in \bR^{p_{\cK} \times p_{\cK}}$, $G \in \bR^{p_{\cK} \times m}$, and $H \in \bR^{p \times p_{\cK}}$ such that
\begin{subequations}\label{cond:LM}
\begin{align}
K A = P K,
&\quad
\offd{P} \ge 0
\label{condA:LM}\\
G= K B,
&\quad
G \ge 0
\label{condB:LM}\\
C = H K,
&\quad
H \ge 0.
\label{condC:LM}
\end{align}
\end{subequations}
\end{secthm}

\begin{proof}
This is a special case of Theorem~\ref{thm:NM}.
\end{proof}

The properness (more precisely, simplicity) of $\cK$ implies ${\rm rank} \; K = n$.
This means that $z = K x$ is an embedding.
In the $z$-coordinates, \eqref{cond:LM} yields
\begin{align}
\Sigma^P: 
\left\{\begin{array}{r@{}l}
\dot z &{}= P z + G u \\
y &{}= H z.
\end{array}\right.
\end{align}
Note that, starting from a $\mathcal{K}$-positive system $\Sigma$, the embedded system $\Sigma^P$ is an $\bR_+^{p_{\cK}}$-positive system, 
as it satisfies \eqref{cond:PR} for $A=P$, $B=G$, and $C=H$.
Therefore, verifying $\mathcal{K}$-positivity of $\Sigma$ is equivalent to 
finding its $\bR_+^{p_{\cK}}$-positive realization.
Since the mapping $z = K x$ is an embedding, 
the obtained $\bR_+^{p_{\cK}}$-positive realization is not minimal, in general.
Finding an $\bR_+^{p_{\cK}}$-positive realization has been investigated in \cite{Farina:95, Hof:97, ADF:96, BF:02, BF:04}. 
In fact, \eqref{cond:LM} can be viewed as a dual of the condition in \cite[Theorem 1]{ADF:96} or \cite[Theorem 5]{BF:02}.

The following theorem takes advantage of the embedding to derive stability conditions for 
$\cK$-positive closed system.
\begin{secthm}[exponential stability]\label{thm:ES}
Let $\cK$ be a simple cone.
A $\cK$-positive linear system $\dot x = Ax$ is 
exponentially stable if there exist $P$ satisfying \eqref{condA:LM} and $v \in \bR^{p_{\cK}}$ such that
\begin{subequations}\label{cond:ES}
\begin{align}
v &> 0\\
- v^\top P &> 0.
\end{align}
\end{subequations}
\end{secthm}
\begin{proof}
This is a special case of Theorem \ref{thm:IES}.
\end{proof}

In literature, the problem of finding an $\bR_+^{p_{\cK}}$-positive realization for a system and the stability problem are typically studied in autonomy.
The novelty of Theorem \ref{thm:ES} is to combine these two properties, taking advantage of \eqref{condA:LM}
to enforce additional linear constraints \eqref{cond:ES} to certify stability. Note that \eqref{cond:ES} 
is strictly related with the stability conditions for $\bR_+^{p_{\cK}}$-positive linear systems of \cite{Rantzer:15,HCH:10}. 

When $\cK$ is proper, the conditions of the theorem can be easily justified by using the embedding.
Without, Theorem~\ref{thm:ES} can be proven by employing
the linear Lyapunov function $V (x) = v^\top K x$, $x \in \cK$.
It follows from \eqref{cond:ES} that $V(x) > 0$ and
\begin{align}\label{cond:ESK}
\dot V(x) = v^\top K A x = v^\top P  K x < 0
\end{align}
for any $x \in \cK \setminus \{0\}$.
This approach is detailed in our preliminary work \cite{KF:21}, which does not rely on embeddings.

\subsection{Positivity and Dissipativity}

For open systems, similar conditions can be derived in terms of dissipativity analysis.
In particular, the notion of exponential dissipativity for cooperative systems \cite[Definition 5.3]{HCH:10} 
provides a useful framework for the analysis of networks \cite[Theorem 5.5]{HCH:10}. 
The following definition extends the
notion of exponential dissipativity to $\cK$-positive systems.

\begin{secdefn}\label{def:ED}
A $\cK$-positive open linear system $\Sigma$ is \emph{exponentially dissipative}
with respect to the supply rate $s(u,y) = r^\top u - q^\top y$, $(r, q) \in \bR^m \times \bR^p$, if
there exists $0 <\lambda \in \bR $ and $ 0 < v \in \bR^n$ such that
\begin{align}\label{eq:ED}
e^{\lambda t_2} v^\top K x(t_2) 
\ \le \ & e^{\lambda t_1} v^\top K x(t_1) \nonumber\\
& + \int_{t_1}^{t_2} e^{-\lambda t} (r^\top u(t) - q^\top y(t)) dt,
\end{align}
for all $t_2 \geq t_1 \geq 0$ and all $(x(\cdot),u(\cdot),y(\cdot)) \in \Sigma$ that satisfy $(x(t),u(t),y(t)) \in (\cK,  \bR_+^m, \bR_+^p)$ for $t \geq 0$.
\red
\end{secdefn}

Note that for $Kx \geq 0$ for all $x \in \cK$. This implies that, for any vector $v \geq 0$, $v^\top Kx: \mathcal{K} \mapsto \bR_{\geq 0}$ is a storage function
(zero if $x = 0$). As usual in dissipativity theory, Definition \ref{def:ED} relates the variation of the storage along system trajectories 
to the the input/output behavior of the system, via the integral of the supply. 

Using the embedding, we obtain the following conditions for exponential dissipativity.
\begin{secthm}[exponential dissipativity]\label{thm:ED}
A $\cK$-positive open linear system $\Sigma$ is 
exponentially dissipative with respect to a supply rate $r^\top u - q^\top y$
if there exist $P$, $G$, and $H$ satisfying \eqref{cond:LM}
and $v \in \bR^{p_{\cK}}$ such that
\begin{subequations}\label{cond:ED}
\begin{align}
v &> 0\\
- v^\top P - q^\top H &> 0\\
r^\top - v^\top G &\ge 0.
\end{align}
\end{subequations}
\end{secthm}
\begin{proof}
This is a special case of Theorem \ref{thm:EDD}.
\end{proof}

Theorem \ref{thm:ED} opens the way to the analysis of interconnections of 
subsystems that are positive with respect to triplets of simple polyhedral cones $(\cK_i, \bR_+^{m_i}, \bR_+^{p_i})$:
\begin{equation} \label{eq:n}
\Sigma_i: 
\left\{\begin{array}{r@{}l}
\dot x_i &{}= A_i x_i + B_i u_i\\
y_i &{}= C_i x_i
\end{array}\right.
\quad
i=1,\dots, N,
\end{equation}
where $x_i \in \bR^{n_i}$, $u_i \in \bR^{m_i}$, and $y_i \in \bR^{p_i}$.
Interconnections are represented by
\begin{align}\label{eq:nc}
u_i = \sum_{j = 1}^N W_{i,j}  \, y_j, \quad i =1,\dots,N.
\end{align}

The exponential stability of the network \eqref{eq:n}, \eqref{eq:nc} 
can be characterized in terms of the exponential dissipativity of
its components and of their interconnection, as clarified by the following
theorem.
\begin{secthm}[network stability]\label{thm1:net}
Let the open linear systems $\Sigma_i$ be $\cK_i$-positive, for $i = 1,\dots,N$, 
where $\cK_i$, $i=1,\dots,N$ are all simple.
Suppose that \vspace{1.5mm}
\begin{itemize}
\renewcommand{\theenumi}{\Roman{enumi}}
\item[(i)] each $\Sigma_i$ satisfies Theorem \ref{thm:ED} for $(v_i, r_i, q_i)$; \vspace{2mm}
\item[(ii)]  $W_{i,j} \ge 0$ for all $i,j = 1, \dots, N$;
\item[(iii)]  $q_j^\top - \sum\nolimits\limits_{i=1}^N r_i^\top  W_{i, j} \ge 0$, for all $j= 1,\dots, N$.
\end{itemize}
Then the network \eqref{eq:n}, \eqref{eq:nc} is
$(\cK_1 \times \cdots \times \cK_N)$-positive and exponentially stable. 
\end{secthm}

\begin{proof}
This is a special case of Theorem \ref{thm3:net}.
\end{proof}

Condition (i) guarantees exponential dissipativity. Condition (ii) guarantees that the interconnection of positive subsystems leads
to a positive network. Finally, Conditions (iii) allows to build a decaying Lyapunov
function for the network, as the combination of the storages of the subsystems. This guarantees network stability.

\section{Monotone Nonlinear Systems}\label{sec:MN}

\subsection{Monotonicity and Incremental Stability}\label{sec:DC}
The main theorems of Section \ref{sec:ML} can be extended to 
monotone nonlinear systems by taking advantage of differential analysis \cite{FS:16,FS:14cdc} and contraction theory \cite{LS:98,FS:14}.
These systems  are characterized by the property that their trajectories preserve a partial order, as clarified by Definition \ref{def:monotonicity}.

Consider the nonlinear system described by
\begin{align} \label{eq:nonlinsys}
\bar \Sigma: 
\left\{\begin{array}{r@{}l}
\dot x &{}= f(x) + B u\\
y &{}= C x,
\end{array}\right.
\end{align}
where $x \in \bR^n$, $u \in \bR^m$, $y \in \bR^p$, and $f: \bR^n \to \bR^n$ is of class $C^2$. 
Recall that any simple cone $\cK$ induces a partial order $\preceq_{\cK}$ given by
\begin{align*}
x_1\preceq_\cK x_2 \quad \mbox{iff} \quad x_2 - x_1 \in \cK.
\end{align*}
This is the original definition of monotonicity  \cite{AS:03,Smith:08}.
\begin{secdefn} \label{def:monotonicity}
The nonlinear system~$\bar \Sigma$ is \emph{monotone} with respect to
the partial order induced by the triplet of  polyhedral cones $(\cK, \bR_+^m, \bR_+^p)$
if
\begin{align}
\begin{array}{c}
x(0) - x'(0)  \in \cK, \; u(t) - u'(t) \in  \bR_+^m, \ \forall t \! \ge \! 0 \\
\implies  \\
x(t) - x'(t) \in \cK, \; y(t) - y'(t) \in \bR_+^p,	\ \forall t \!\ge \! 0 \ ,
\end{array}
\end{align}
for all $(x(\cdot), u(\cdot), y(\cdot)) \in \bar \Sigma$, $(x'(\cdot), u'(\cdot), y'(\cdot)) \in \bar \Sigma$.

For a closed system $\dot x = f(x)$, monotonicity simply means that
$x(0) - x'(0)  \in \cK  $ implies $ x(t) - x'(t) \in \cK$ for all $t \ge 0$,
along any state trajectory $x(\cdot )$. 
\red
\end{secdefn}

In what follows we will adopt the simpler terminology of $\mathcal{K}$\emph{-monotone systems}
for the systems in Definition \ref{def:positivity}, leaving the input and output partial orders implicit.

Monotonicity is strictly connected to positivity.
In fact monotonicity of a nonlinear system is equivalent to 
positivity of its variational system, as clarified in Theorem \ref{thm:mbp} below.
Consider the variational system of $\bar \Sigma$ along $(x(\cdot),u(\cdot),y(\cdot))$:
\begin{align}
d\bar \Sigma: 
\left\{\begin{array}{r@{}l}
\dot {\delta x} &{}= \partial f(x(t)) \delta x + B \delta u\\
\delta y &{}= C \delta x,
\end{array}\right.
\end{align}
where $\delta x \in \bR^n$, $\delta u \in \bR^m$, $\delta y \in \bR^p$, and $\partial f(x)$ denotes the Jacobian of $f$ computed at $x$. 
\begin{secthm}[monotonicity and positivity] \label{thm:mbp}
The open nonlinear system~$\bar \Sigma$ is $\cK$-\emph{monotone} 
if and only if its variational system $d \bar \Sigma$ satisfies
\begin{align}\label{def:NM}
\begin{array}{c}
\delta x(0)  \in \cK, \; \delta u(t)  \in  \bR_+^m,	\ \forall t \! \ge \! 0 \\
\implies \\
 \delta x(t) \in \cK, \; \delta y(t) \in \bR_+^p,	\ \forall t \!\ge \! 0
\end{array}
\end{align} 
for all $(x(\cdot), u(\cdot), y(\cdot)) \in \bar\Sigma$, $(\delta x(\cdot), \delta u(\cdot), \delta y(\cdot)) \in d\bar{\Sigma}$.

For a closed system $\dot x = f(x)$, monotonicity simply means that
$\delta x(0)  \in \cK$ implies $ \delta x(t) \in \cK$ for any $t \ge 0$
along any bounded state trajectory $x(\cdot )$.
\red
\end{secthm}
\begin{proof}
The theorem corresponds to \cite[Proposition 1.5]{Walcher:01} and to \cite[Theorem~1]{FS:16}.
\end{proof}

Theorem \ref{thm:mbp} bridges monotonicity and positivity, opening the way
to the extension of the results of Section \ref{sec:ML} to nonlinear systems. The first
of these extensions provides conditions to verify monotonicity.
\begin{secthm}[monotonicity]\label{thm:NM}
An open nonlinear system~$\bar \Sigma$ is $\cK$-monotone 
if and only if 
there exist $P: \bR^n \to \bR^{p_{\cK} \times p_{\cK}}$, $G \in \bR^{p_{\cK} \times m}$, and $H \in \bR^{p \times p_{\cK}}$ such that, 
for all $x\in \bR^n$, 

\begin{subequations}\label{condf:NM}
\begin{align}
K \partial f(x) &= P(x) K,
\quad
\offd{P(x)} \ge 0
\label{condA:NM}\\
G&= K B,
\; \qquad
G \ge 0
\label{condB:NM}\\
C &= H K,
\;\;\;\;\quad
H \ge 0.
\label{condC:NM}
\end{align}
\end{subequations}
\end{secthm}
\begin{proof}
The proof is in Appendix \ref{proof:NM}.
\end{proof}

The reader will notice that \eqref{condB:NM} and \eqref{condC:NM} correspond to \eqref{condB:LM} and \eqref{condC:LM},
and that \eqref{condA:NM} adapts \eqref{condA:LM} to the variational dynamics, replacing the state matrix $A$ with 
the Jacobian of $f$ computed at each point of the system state space.
For a closed system $\dot{x} = f(x)$, \eqref{condA:NM} is a necessary and sufficient condition for 
its monotonicity. Its dual condition can be found in \cite[Lemma 2]{KF:20}. 

Theorem \ref{thm:mbp} shows how the property of monotonicity for a nonlinear system
can be established by looking at its variational dynamics, as shown by Theorem \ref{thm:NM}.
A similar approach can be pursued to derive conditions for incremental stability  
of monotone systems, using contraction theory \cite{Angeli:02,FS:14}. Indeed, Theorem 
\ref{thm:IES} provides scalable conditions for incremental stability based on the variational system. 
As in the linear case, scalability follows from the fact that, 
for $\cK$-monotone systems, \eqref{condf:NM} leads to the embedding 
$\delta z = K \delta x$ in variational coordinates, which 
yields the positive variational system $d\bar \Sigma^P$
\begin{align}
d\bar \Sigma^P: 
\left\{\begin{array}{r@{}l}
\dot {\delta z} &{}= P(x(t)) \delta z + G \delta u\\
\delta y &{}= H \delta z \ .
\end{array}\right.
\end{align}

We briefly recall the notion of incremental stability \cite{Angeli:02}.
\begin{secdefn}\label{def:ies}
The closed system $\dot x = f(x)$, $x\in \bR^n$, is \emph{incrementally exponentially stable} 
if there exist $1 \leq k \in \bR$ and $0 < \lambda \in \bR$ such that 
\begin{align}
| x(t) - x'(t) |_1 \le k e^{-\lambda t} |x(0) - x'(0)|_1
\end{align}
for all trajectories $x(\cdot)$ and $x'(\cdot)$ and all $t \ge 0$.
\red
\end{secdefn}

In this paper, incremental exponential stability is defined by using the vector $1$-norm. However, this is not essential, and the property is equivalent to any vector $i$-norm, $i\in [1, \infty]$.

Below, we provide a condition for incremental exponential stability.

\begin{secthm}[incremental exponential stability]\label{thm:IES}
A closed $\cK$-monotone system $\dot x = f(x)$  is incrementally exponentially stable if
there exists $P(\cdot)$ satisfying \eqref{condA:NM}, $0 < \lambda \in \bR$, and $v \in \bR^{p_{\cK}}$ 
such that, for all $x \in \bR^n$, 
\begin{subequations}\label{cond:IES}
\begin{align}
v &> 0
\label{cond1:IES}\\
-v^\top P(x) &\ge \lambda v^\top.
\label{cond2:IES}
\end{align}
\end{subequations}
\end{secthm}
\begin{proof}
The proof is in Appendix \ref{proof:IES}.
\end{proof}

The novelty of Theorem \ref{thm:IES} is in the combination of the condition for monotonicity \eqref{condA:NM}
and of conditions for contraction \eqref{cond:IES}, which make use of the embedding in variational coordinates.
The reader will notice that the conditions of Theorem \ref{thm:IES} are a direct extension of the conditions of Theorem \ref{thm:ES}.
The incremental exponential stability property of Definition \ref{def:ies} also guarantees the existence of a globally exponentially 
stable equilibrium point $x_e$ for the system $\dot{x} = f(x)$. 
This follows from the Banach fixed-point theorem, since $(\bR^n,|\cdot|_1)$ is a complete metric space. 
All trajectories of $\dot{x} = f(x)$ are also bounded, since $| x(t) - x_e |_1 \le k |x(0) - x_e|_1$.
See also \cite[Theorem 2.2]{Kawano:22}.

When $\cK$ is proper, a weaker incremental stability condition than Theorem 4.5 can be derived as follows
\begin{align*}
v^\top K &> 0 \\
-v^\top P(x) K &\ge \lambda v^\top K.
\end{align*}
We later develop numerical algorithm for finding $v$, $P(x)$, and $K$ simultaneously.  If we use this condition, we need to handle multiplications of three unknown variables. In contrast, with Theorem 4.5, we only have to deal with multiplications of two unknown variables.

As another generalization, one can relax $\cK$-monotonicity. In fact, Theorem~\ref{thm:IES} yields the following corollary.
\begin{seccor}\label{cor:IES}
A closed system $\dot x = f(x)$ is incrementally exponentially stable if there exist $K \in \bR^{p_{\cK} \times n}$ with ${\rm rank} \; K = n$, $P: \bR^n \to \bR^{p_{\cK} \times p_{\cK}}$, $0 < \lambda \in \bR$, and $v \in \bR^{p_{\cK}}$ 
such that  for all $x \in \bR^n$, $K \partial f (x) = P(x) K$ and 
\begin{subequations}\label{cond:IES2}
\begin{align}
v &> 0\\
-v^\top |\offd{P}(x)| &\ge \lambda v^\top.
\end{align}
\end{subequations}
\end{seccor}
\begin{proof}
The proof is in Appendix \ref{proof:IES_cor}.
\end{proof}

\subsection{Monotonicity and (Differential) Dissipativity} 
Mirroring the linear case, we develop here a suitable notion of
dissipavitity that will later use for network analysis. 

\begin{secdefn}\label{def:EDD}
A $\cK$-monotone  open nonlinear system $\bar \Sigma$ is \emph{exponentially differentially dissipative}
with respect to the (differential) supply rate $s(\delta u, \delta y) = r^\top \delta u - q^\top \delta y$, $(r, q) \in \bR^m \times \bR^p$, 
if there exist $0< \lambda \in \bR$ and $0 < v \in \bR^{p_{\cK}}$ such that
\begin{align}\label{eq:EDD} 
e^{\lambda t_2} v^\top K \delta x(t_2)  \le \
& e^{\lambda t_1} v^\top K \delta x(t_1) \nonumber \\
&  + \int_{t_1}^{t_2} e^{\lambda t} (r^\top \delta u(t) - q^\top \delta y(t)) dt
\end{align}
for all $t_2 \geq t_1 \geq 0$ and all $(x(\cdot),u(\cdot),y(\cdot)) \in \bar\Sigma$, $(\delta x(\cdot),\delta u(\cdot),\delta y(\cdot))\in d\bar\Sigma$, 
that satisfy $(\delta x(t),\delta u(t),\delta y(t)) \in (\cK,  \bR_+^m, \bR_+^p)$ for $t \geq 0$.
\red
\end{secdefn}

\eqref{eq:EDD} adapts \eqref{eq:ED} to the variational system $d\bar\Sigma$. 
\eqref{eq:EDD} is equivalent to the condition: 
\begin{align}\label{eq2:EDD}
v^\top K \dot {\delta x} + \lambda v^\top K \delta x \le r^\top \delta u - q^\top \delta y \ ,
\end{align}
whose left-hand side is related to the incremental exponential stability condition \eqref{cond:IES}.

\eqref{eq:EDD} is also equivalent to the following \emph{incremental exponential dissipativity} condition, stated in terms of pairs of trajectories of $\bar\Sigma$:
\begin{align}
& e^{\lambda t_2} v^\top K (x(t_2) - x'(t_2)) 
\le \ e^{\lambda t_1} v^\top K (x(t_1) - x'(t_1)) \nonumber\\
& \quad + \int_{t_1}^{t_2} e^{\lambda t} (r^\top (u(t) - u'(t)) - q^\top (y(t) - y'(t))) dt
\end{align}
for all $t_2 \geq t_1 \geq 0$ and all $(x(\cdot), u(\cdot), y(\cdot)) \in \bar\Sigma$, $(x'(\cdot), u'(\cdot), y'(\cdot))\in \bar\Sigma$ such that
 $x(t) - x'(t) \in \cK$, $u(t) - u'(t) \in  \bR_+^m$, and $y(t) - y'(t)  \in  \bR_+^p$ for all $t \ge 0$.

Mimicking Theorem \ref{thm:ED}, we have the following dissipativity conditions.

\begin{secthm}[exponential differential dissipativity]\label{thm:EDD}
A $\cK$-monotone nonlinear system $\bar \Sigma$ is 
exponentially differentially dissipative with respect to 
a supply rate $r^\top \delta u - q^\top \delta y$ 
if there exist $P(\cdot)$, $G$, and $H$ satisfying \eqref{condf:NM},
and $0 < \lambda \in \bR$ and $v \in \bR^{p_{\cK}}$ such that
\begin{subequations}\label{cond:EDD}
\begin{align}
v &> 0
\label{cond1:EDD}\\
- v^\top P(\cdot) - q^\top H &\ge  \lambda v^\top, 
\quad
\forall x \in \bR^n
\label{cond2:EDD}\\
r^\top - v^\top G & \ge 0.
\label{cond3:EDD}
\end{align}
\end{subequations}
\end{secthm}
\begin{proof}
The proof is in Appendix \ref{proof:EDD}.
\end{proof}

We observe that for $q=0$, \eqref{cond1:EDD} and \eqref{cond2:EDD} correspond to \eqref{cond:IES}.
In fact, for $q \ge 0$, since $H \ge 0$, exponential differential dissipativity entails
incremental exponential stability (for closed sytems).

Theorem \ref{thm:EDD} allows for the analysis of interconnections of
open nonlinear systems that are monotone with respect to the partial order induced by triplets of simple polyhedral cones $(\cK_i, \bR_+^{m_i}, \bR_+^{p_i})$
\begin{align} \label{eq:nss}
\bar \Sigma_i: 
\left\{\begin{array}{r@{}l}
\dot x_i &{}= f_i (x_i )+ B_i u_i\\
y_i &{}= C_i x_i
\end{array}\right.
\quad
i = 1, \dots, N,
\end{align}
where~$x_i \in \bR^{n_i}$,~$u_i \in \bR^{m_i}$, and $y_i \in \bR^{p_i}$, and 
$f_i: \bR^{n_i} \to \bR^{n_i}$ is of class $C^2$. 
We also consider nonlinear interconnections of the form
\begin{align}\label{eq:n_int_ss}
u_i = W_i( y_1,\dots, y_N) , \quad i = 1,\dots,N,
\end{align}
where $W_i: \bR^{p_1} \times \cdots \times \bR^{p_N} \to \bR^{m_i}$ is of class $C^2$.

Theorem \ref{thm1:net} is generalized as follows.
\begin{secthm}[network incremental stability, out-degree]\label{thm3:net}
Consider $\cK_i$-monotone subsystems $\bar \Sigma_i$, for $i = 1,\dots,N$,
where $\cK_i$, $i=1,\dots,N$ are all simple cones.
Suppose that
\begin{itemize}
\renewcommand{\theenumi}{\Roman{enumi}}
\item[(i)] each $\Sigma_i$ satisfies Theorem \ref{thm:EDD} for $(v_i, r_i, q_i)$; \vspace{2mm}
\item[(ii)]  $\partial_{y_j} W_i  (\cdot ) \ge 0$ for all $i, j = 1,\dots,N$;
\item[(iii)]  $q_j^\top - \sum\nolimits\limits_{i=1}^N r_i^\top  \partial_{y_j} W_i  (\cdot) \ge 0$, for all $j = 1,\dots, N$.
\end{itemize}
Then, the network \eqref{eq:nss}, \eqref{eq:n_int_ss} is
$(\cK_1 \times \cdots \times \cK_N)$-monotone and incrementally exponentially stable. 

Furthermore, whenever (i) holds for $r_i^\top = v_i^\top G_i$, 
incrementally exponentially stability is guaranteed even if (ii) is replaced by 
\begin{itemize}
\item [(iia)] $\offd{P_i(\cdot) + G_i  \partial_{y_i} W_i  (\cdot) H_i}  \! \ge \!  0$, for all $i= 1,\dots, N$;
\item [(iib)] $G_i  \partial_{y_j} W_i (\cdot) H_j \ge 0$ for all $j$ such that $i \neq j$.
\end{itemize}
\end{secthm}
\begin{proof}
The proof is in Appendix \ref{proof3:net}.
\end{proof}

Theorem~\ref{thm3:net} can be utilized for modular control design of networks.
Given the interconnection rules $W_i  (y_1, \dots, y_N)$ satisfying Condition (ii), 
we first computes the pairs $(r_i, q_i)$, $i=1,\dots,N$ to satisfy Condition (iii).
Then, we develop controllers for each subsystem that guarantee Condition (i).
We will explore this approach in Section \ref{sec:Ex}.

Note that (iii) of Theorem \ref{thm3:net} is described in terms of the out-degree of each node, weighted by $r_j$.
Similar conditions can be derived using the weighted in-degree.
\begin{secthm}[network incremental stability, in-degree]\label{thm4:net}
Consider the $\cK_i$-monotone subsystems $\bar \Sigma_i$, $i = 1,\dots,N$,
where $\cK_i$, $i=1,\dots,N$ are all simple cones.
Suppose that
\begin{itemize}
\item[(i)] for each $\bar \Sigma_i$, $i = 1,\dots,N$, 
there exist $P_i(\cdot)$, $G_i$, and $H_i$ satisfying
 \eqref{condf:NM}, and $0 < \lambda_i \in \bR$, $w_i \in \bR^{p_{\cK_i}}$, $r_i \in \bR^{m_i}$, and $q_i \in \bR^{p_i}$ such that
\begin{subequations}
\begin{align}
w_i &> 0\\
- P_i(\cdot) w_i - G_i q_i &\ge \lambda_i w_i\\ 
r_i - H_i w_i &\ge 0 ;
\end{align}
\end{subequations}
\item[(ii)] $\partial_{y_j} W_i  (\cdot) \ge 0$ for all $i, j = 1,\dots,N$;
\item[(iii)] $q_i - \sum\nolimits\limits_{j=1}^N \partial_{y_j} W_i  (\cdot ) r_j \ge 0$, for all $i= 1,\dots, N$.
\end{itemize}
Then the network \eqref{eq:nss}, \eqref{eq:n_int_ss} is
$(\cK_1 \times \cdots \times \cK_N)$-monotone and incrementally exponentially stable.

Furthermore, whenever (i) holds for $r_i = H_iw_i$, 
incrementally exponentially stability is guaranteed even if (ii) is replaced by 
(iia) and (iib) of Theorem~\ref{thm3:net}.
\end{secthm}
\begin{proof}
The proof is in Appendix \ref{proof4:net}.
\end{proof}

Similarly to Corollary~\ref{cor:IES}, Theorems~\ref{thm3:net} and~\ref{thm4:net} can be generalized by relaxing the requirement $\offd{P}(\cdot) \ge 0$.

\section{Example: Nonlinear Mass-Spring System}\label{sec:Ex}
\label{sec:ex}
We develop the analysis of a nonlinear mass-spring system using 
the theoretical results of the paper. The system is represented
by the following equations
\begin{align}\label{sys:nmech}
\left\{\begin{array}{r@{}l}
\begin{bmatrix}
\dot x_1 \\ \dot x_2
\end{bmatrix}
&{}=
\begin{bmatrix}
x_2 \\
- k ( x_1)
\end{bmatrix}
+
\begin{bmatrix}
0 \\ 1
\end{bmatrix} u \\[2.5mm]
y &{}= \begin{bmatrix}
1 & 0
\end{bmatrix} x,
\end{array}\right.
\end{align}
where~$\underline k \leq \partial k(x_1) \leq \overline k$ for all $x_1$ and $x = \begin{bmatrix}x_1 \!&\! x_2\end{bmatrix}^\top$.
 
We again use the cone $\cK$ defined by $K_1$ and $K_2$ in \eqref{cone:mech}
and a linear state-feedback controller given by 
\begin{equation}
u = F_1 x_1 + F_2 x_2.
\label{sys:nfeedback}
\end{equation}

From Theorem \ref{thm:NM}, \eqref{condA:NM},
the closed-loop system \eqref{sys:nmech},\eqref{sys:nfeedback} is $\cK$-monotone
 if and only if there exists $P:\bR^2 \to \bR^{2 \times 2}$ such that
\begin{align*}
&\begin{bmatrix}
1 & 0 \\
2 & 1
\end{bmatrix}
\begin{bmatrix}
0 & 1 \\
F_1 - \partial k (x_1) & F_2
\end{bmatrix}
=
\underbrace{\begin{bmatrix}
P_{11}(x) & P_{12}(x) \\
P_{21}(x) & P_{22}(x)
\end{bmatrix}}_{P(x)}
\begin{bmatrix}
1 & 0 \\
2 & 1
\end{bmatrix}
\end{align*} 
with $ P_{12}(x) \ge 0,  P_{21}(x) \geq 0$, for all $x\in\bR^2$.
This leads to the condition
\begin{align}\label{cond2:mech_cp}
-4 -  \partial k(x_1) + F_1 - 2F_2 \ge -4 -  \overline k + F_1 - 2F_2 \ge 0
\end{align}
since
\begin{align*}
P(x) 
=
\begin{bmatrix}
-2 & 1\\
-4 - \partial k(x_1) + F_1 - 2 F_2 & 2 + F_2
\end{bmatrix}.
\end{align*}
Also, from \eqref{condB:NM} and \eqref{condC:NM} we get 
\begin{align*}
G := \begin{bmatrix}
0 \\ 1
\end{bmatrix},
\qquad
H := \begin{bmatrix}
1 & 0
\end{bmatrix}.
\end{align*}

We now consider Theorem \ref{thm:EDD}, and specifically \eqref{cond:EDD},
to show that the system is exponentially differentially dissipative. 
For some $\lambda > 0$, we need to find $v = \begin{bmatrix} v_1  \!&\! v_2\end{bmatrix}^{\!\top}$ such that
\begin{subequations} \label{cond:EDD_ex0}
\begin{align}
v &> 0\\
-q + 2 v_1 + v_2 ( 4 + \partial k(x_1) - F_1 + 2 F_2) &> 0\\
-v_1 - v_2 (2 + F_2) &> 0\\
r &\ge v_2.
\end{align}
\end{subequations}
Note that $-q + 2 v_1 + v_2 ( 4 + \partial k(x_1) - F_1 + 2 F_2) > 0$ can be replaced by the lower bound 
$-q + 2 v_1 + v_2 ( 4 + \underline{k} - F_1 + 2 F_2)>0$. Furthermore, 
$v_2$ can be chosen as $1$ without loss of generality.
This can be confirmed by dividing each inequality by $v_2$ and 
introducing new variables $\bar v_1:=v_1/v_2$, $\bar q:=q /v_2$, and $\bar r:=r/v_2$. 
For~$v_2=1$, we chose the smallest~$r$ satisfying the condition, namely $r=1$. 

Combining  \eqref{cond2:mech_cp}, \eqref{cond:EDD_ex0}, and the
observations above, we get
\begin{subequations}\label{cond:EDD_ex}
\begin{align}
v_1 &> 0
\label{cond1:EDD_ex}\\
- q + 2 v_1 + 4 + \underline k > F_1 - 2 F_2 & > 4  +  \overline k
\label{cond2:EDD_ex}\\
 - F_2 &> v_1 + 2.
\label{cond3:EDD_ex}
\end{align}
\end{subequations}
\eqref{cond:EDD_ex} is feasible for arbitrary $q$, $\underline k$, and $\overline k$.
To see this, note that there exists~$v_1 > 0$ such that
$
- q + 2 v_1 + 4 + \underline k >  4  +  \overline k
$.
Then one can set $F_2$ to satisfy \eqref{cond3:EDD_ex}. $F_1$ can be finally used to satisfy \eqref{cond2:EDD_ex}. 
Any solution to \eqref{cond:EDD_ex} guarantees that the closed-loop system \eqref{sys:nmech},\eqref{sys:nfeedback} 
is exponentially differentially dissipative, thus 
$\cK$-monotone and incrementally exponentially stable for $q\geq0$.

For simplicity, let us consider the case $q = 1$, $\underline k=-1$, and $\overline k=1$.
This corresponds to a nonlinear spring with attractive and repulsive action, like in simple models of buckling.
In such a case, $F_1=-6$, $F_2 = -6$, and $v_1 = 3$ satisfy \eqref{cond:EDD_ex}.
Fig.~\ref{fig:mech} shows the phase portrait of the closed-loop system \eqref{sys:nmech},\eqref{sys:nfeedback} 
for the nonlinear spring $k(x_1) = (\sin(x_1) - x_1)/2$, which satisfies  $-1 \leq \partial k(x_1) \leq 1$ for all $x_1$.
The origin is a globally exponentially stable equilibrium point.

\begin{figure}[tb]
\begin{center}
\includegraphics[width=0.68\columnwidth]{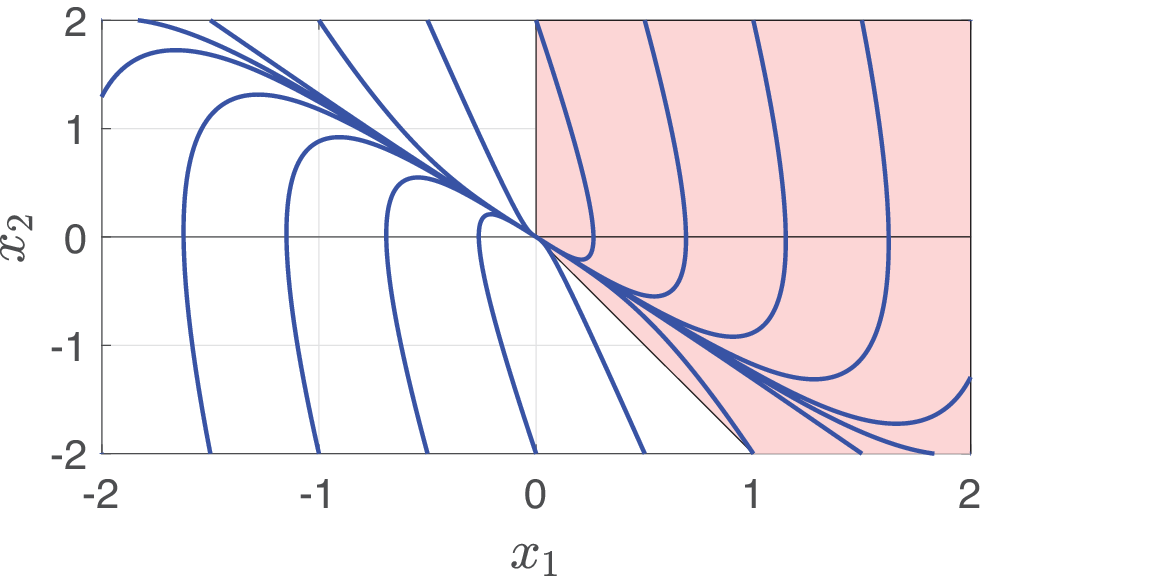} 
\caption{Phase portrait of the closed-loop subsystem {\eqref{sys:nmech} with $u = \begin{bmatrix}-6 \!&\! -6\end{bmatrix}x$}. The red region represents the cone $\cK$} 
\label{fig:mech}
\end{center}
\end{figure}

We now consider a network of the form \eqref{eq:nss},\eqref{eq:n_int_ss}, where each subsystem is a 
controlled nonlinear mass-spring system \eqref{sys:nmech},\eqref{sys:nfeedback}. From Theorem~\ref{thm3:net}, any solution to \eqref{cond:EDD_ex} guarantee that 
a network with arbitrary structure and arbitrary size $N$ is incrementally
exponentially stable if
\begin{subequations} \label{eq:net_example}
\begin{align}
\partial_{y_i} W_j ( \cdot ) \ge 0, \quad& \forall i,j = 1,\dots,N \label{eq:net_example_int}\\
\displaystyle q \ge \sum_{i = 1}^N \partial_{y_i} W_j  ( \cdot ), \quad& \forall j=1,\dots,N. \label{eq:net_example_strength}
\end{align}
\end{subequations}

We consider a network of heterogeneous subsystems. Specifically, the nonlinear spring for the $i$-th system satisfies
$k_i(x_{i,1}) = a_i( \sin(x_{i,1}) - x_{i,1})$ for randomly generated $a_i \in [0, 1/2]$.
This guarantee that $-1 \leq \partial k_i(x_{i,1}) \leq 1$ for all $x_{i,1}$. 
We choose $N=100$ and consider constant interconnection weights, i.e., 
$u_i = \sum_{j=1}^N W_{i,j} x_j$, for randomly generated $0 \leq W_{i,j} \leq 1/N$.
This range guarantees that the network interconnections satisfy \eqref{eq:net_example}.
Thus, from Theorem~\ref{thm3:net}, the network is $\cK$-monotone and incrementally exponentially stable. 

Figure~\ref{fig:net} shows the convergence of the aggregate output of the network $y$
from generic initial positions $|y(0)|_{\infty} \leq 1$ and initial velocities $|\dot y(0)|_{\infty} \leq 1$, where $| \cdot |_\infty$ denotes the vector $\infty$-norm.
For different inputs, independently on the initial condition, the each network component converges to steady-state.

\begin{figure}[tb]
\begin{center}
\includegraphics[width=0.8\columnwidth]{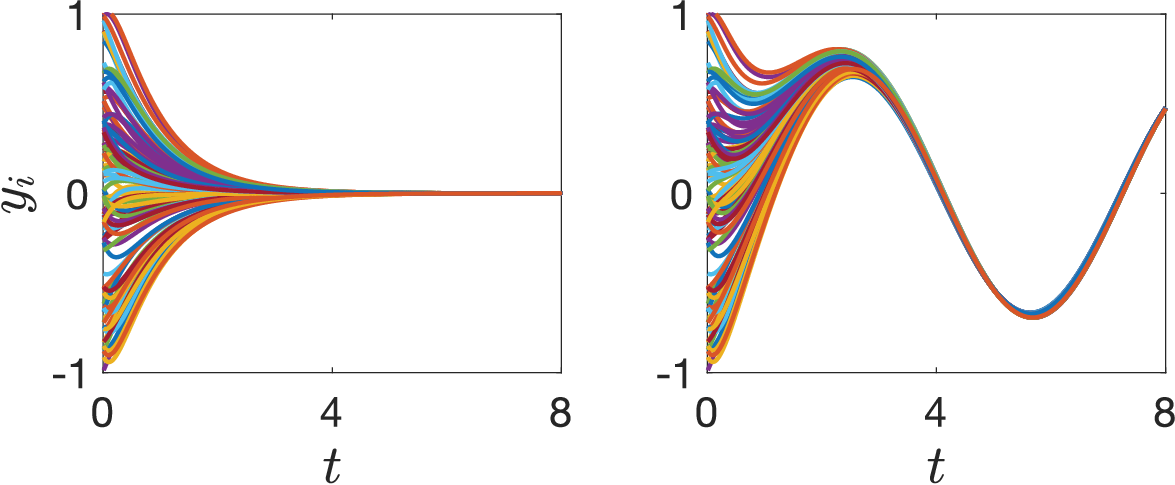} 
\caption{Closed-loop trajectories of the network interconnection of  \eqref{sys:nmech} with $u_i = \begin{bmatrix}-6 \!&\! -6\end{bmatrix} x_i +  \sum_{j=1}^N W_{i,j} y_j + \hat u_i$, for randomly generated $0 \leq W_{i,j} \leq 1/N$, where $N=100$:
$\hat u_i = 0$ (left) and $\hat u_i =  5\sin (t)$ (right), for $i = 1,\dots,N$. } 
\label{fig:net}
\end{center}
\end{figure}

\section{Co-Design of Cone and Controller}\label{sec:AC}

The example in Section \ref{sec:ex} illustrates the potential 
of the theory for analysis and synthesis of feedback systems and networks. 
For a given simple polyhedral cone $\cK$ the conditions
of our theorems lead to linear programs (LP), which can be easily solved. 
However, Section \ref{sec:ME} shows how the
selection of the ``wrong'' cone may lead to infeasibility even if 
stability and monotonicity are both achievable in closed loop. In what follows
we start developing a numerical solution to this problem. We propose an
algorithm that co-design the cone $\cK$ and the feedback controller $u= Fx + \bar u$, with the goal of enforcing $\cK$-monotonicity together with incremental exponential stability or exponential differential dissipativity.
The algorithm is inspired to the earlier attempt \cite{KF:20}.

\subsection{Co-Design Monotonicity and Incremental Stability}
We consider the co-design of a cone $\cK$, vector $v$, and state feedback $F$ to satisfy Theorem \ref{thm:IES}.
Recall that $\cK$ is a simple cone if and only if ${\rm rank} \; K=n$.

\begin{secprob}\label{prob:IES}
For the closed-loop system 
\begin{align}\label{sys:cl}
\dot x = f(x) + B F x \qquad x\in \bR^n
\end{align}
find a matrix $K$ with ${\rm rank} \; K =n$ and a state feedback $F$ such that the closed-loop system
is $\cK$-monotone and incrementally exponentially stable. 
\red
\end{secprob}

To make the problem tractable, we assume that there exists 
a finite set of matrices $\cA := \{ A_1, \dots, A_L \}$ such that
\begin{equation} 
\label{eq:cvx}
\partial f(x) \in \mbox{Convex-hull} (\cA), 
\quad 
\forall  x \in \bR^n,
\end{equation}
that is, for any $x$, there exist weights $\theta_i(x) \geq 0$, $i = 1,\dots,L$ such that 
\begin{equation}
\label{eq:cvx2}
\partial f(x) = \sum_{i=1}^L \theta_i(x) A_i \ ,\qquad \sum_{i=1}^L \theta_i(x) = 1. 
\end{equation}

The following proposition is instrumental.
\begin{secprop}\label{prop1:ES_alg}
Take $\varepsilon_1 > 0$ and 
suppose that $f(x)$ admits the convex relaxation \eqref{eq:cvx}.
Given $K$ with ${\rm rank} \; K = n$ and $v > 0$, 
the following LP problem is feasible:
\begin{subequations} \label{opt:c}
\begin{align}
&\max_{F, P_1, \dots, P_L,  c_p, c_s} c \nonumber \\
\mbox{s.t. \; }
&K (A_i + B F) = P_i K \label{opt:cA} \\
&\offd{P_i} \ge c_p \label{opt:mono} \\
&- v^\top P_i \ge c_s  \label{opt:IES}\\
&c_p \ge c, \; c_s \ge c + \varepsilon_1. 
\end{align}
\end{subequations}
where $i = 1,\dots,L$. Then, if $c \ge 0$, the closed-loop system \eqref{sys:cl} is 
$\cK$-monotone and incrementally exponentially stable.
\end{secprop}
\begin{proof}
The proof is in Appendix \ref{proof:ES_alg}.
\end{proof}

If $c \ge 0$, the LP problem \eqref{opt:c} delivers a state feedback controller that 
guarantees $\cK$-monotonicity and incremental exponential stability:
$K$, $v$, and $P(x) = \sum_{i=1}^L \theta_i(x) P_i$ satisfy
the conditions of Theorem \ref{thm:IES}. 
Otherwise, negative $c_p$ and $c_s$ provide a measure of the gap 
to $\cK$-monotonicity and to incremental exponential stability.
This information can be used to search for a new cone $\cK$ and
a new vector $v > 0$, with
the goal of reducing such gap. 

In what follows, we work through a variational approach. 
The algorithm produces a sequence of small updates on 
$K$ and $v$ with the goal of making $c$ increase.
Such small updates must be compatible with the constraints \eqref{opt:c}.
For example, for small variations, \eqref{opt:cA} leads to
\begin{align*}
(K + \delta K)  (A_i + B (F + \delta F))
 = (P_i + \delta P_i) (K + \delta K)
\end{align*}
and ignoring high-order terms we get
\begin{align*}
 \delta K(A_i + BF) + K B \delta F =  P_i \delta K + \delta P_i K \ .
\end{align*}
Given $K$, $v$, $F$, $P_i$ from  \eqref{opt:c},
we can search for small perturbations  $\delta K$, $\delta F$, and $\delta P_i$, 
with the goal of increasing $c$.
This approach leads to the following optimization problem for updating $K$ and $v$:
\begin{subequations} \label{opt2:c}
\begin{align}
&\max_{\delta K, \delta F, \delta v, \delta P_1, \dots, \delta P_L, c_p, c_s} c \nonumber \\
\mbox{s.t. \; }
& \delta K\! (A_i \!+ \!B F ) + K \!B \delta F \!=\!   \delta P_i K\!+\! P_i \delta K\\
&{\rm rank} \; (K + \delta K) = n \\
&\offd{P_i + \delta P_i} \ge c_p \label{opt2:mono} \\
&v^\top +  \delta v^\top \ge \varepsilon_1 \\
&- v^\top P_i - \delta v^\top P_i - v^\top \delta P_i \ge c_s \label{opt2:IES} \\
&c_p \ge c, \; c_s \ge c + \varepsilon_1. 
\end{align}
\end{subequations}
where $i = 1,\dots,L$ and with the additional constraints that $(\delta K, \delta F, \delta v, \delta P_1, \dots,\delta P_L)$ must remain small.

To constrain ``$(\delta K, \delta F, \delta v, \delta P_1, \dots,\delta P_L)$ remain small'' is short notation for a
set of linear constraints of the form $-\varepsilon_2 \leq \delta K \leq \varepsilon_2$ for some small $\varepsilon_2>0$ (and similar for the other variables).
Likewise, the constraint ``${\rm rank} \; (K + \delta K) =n$''  is not a linear constraint, but 
can be enforced by taking $\delta K$ sufficiently small or by using the relaxation
$K + \delta K = (I + Q) K$ for some $Q \in \bR^{p_{\cK} \times p_{\cK}}$.
In fact, if $I + Q$ is non-singular then ${\rm rank} \; (K + \delta K) = {\rm rank} \; K = n$.
For instance, non-singularity is guaranteed under linear constraints $ - 1/(p_\cK +1) \le Q_{i,j} \le 1/(p_\cK +1)$ for all $i,j =1,\dots,p_\cK$.

Using the obtained~$\delta K$ and $\delta v$, we update $K \leftarrow K + \delta K$ and $v \leftarrow v + \delta v$. 
The proposed algorithm is summarized below.
According to the discussions above, \eqref{opt:c} and \eqref{opt2:c} 
are always feasible but they may produce marginal improvements on $c$.
This can be used to terminate the algorithm (fail). 

\begin{algorithm}
\caption{Solution to Problem~\ref{prob:IES}}
\label{alg:IES} 
\begin{algorithmic}[1]
\small
\Require $\cA := \{ A_1, \dots, A_L\}$, $B$, initial $K$ with ${\rm rank} \; K =n$, initial $v > 0$,
$\varepsilon_1>0$, $\varepsilon_2>0$, $\varepsilon_3>0$
\Ensure Solutions $K$ and $F$ to Problem~\ref{prob:IES} or {\bf Fail}
\While {{\bf True}}
\State {Solve \eqref{opt:c}}
\If {$c \ge  0$}
\State {\bf Return $K$ and~$F$}
\EndIf
\State {Solve \eqref{opt2:c}}
\State {$K \leftarrow K + \delta K$} and {$v \leftarrow v + \delta v$}
\If {$c$ improves less than $\varepsilon_3$}
\State {\bf Return Fail}
\EndIf
\EndWhile
\end{algorithmic}
\end{algorithm}

\begin{secrem} \label{rem:fail_and_complexity}
The complexity of the polyhedral cone, that is, the number of rows of $K$ is
decided at the beginning of the program, through the initial matrix $K$.
The algorithm does not increase the number of rows of $K$ throughout its execution.
A ``fail'' may indicate that the number of rows of $K$ needs to be increased to
reach feasibility. However, if the number of rows is too large, the obtained $K$ can be redundant. Namely, to generate the same cone $\cK$, some $K_j$ can be removed.
\red
\end{secrem}

\begin{secrem}
According to Corollary~\ref{cor:IES}, Algorithm 1 can be extended by replacing the set of~\eqref{opt:mono} and~\eqref{opt:IES}  with $-v^\top |\offd{P_i}| \ge c_s$ and the set of~\eqref{opt2:mono}  and~\eqref{opt2:IES} with 
\begin{align*}
&-v^\top |\offd{P_i}| - \delta v^\top |\offd{P_i}| \\
&\qquad - v^\top |\offd{\delta P_i}| \ge c_s.
\end{align*}
As a solution, this new algorithm gives a feedback controller achieving incremental exponential stability.
\red
\end{secrem}

\begin{secrem}\label{rem:OF}
Algorithm \ref{alg:IES} can be used to design dynamic output feedback controllers
\begin{align*}
\left\{\begin{array}{r@{}l}
\dot x_c &{}= A_c x_c + B_c y\\
u &{}= C_c x_c + D_c y.
\end{array}\right.
\end{align*}
This can be achieved by replacing in \eqref{opt:c} and \eqref{opt2:c} each $A_i + B F$ by
\begin{align*}
\begin{bmatrix}
A_i + B D_c C & B C_c\\
B_c C & A_c
\end{bmatrix},
\end{align*}
for all $i=1,\dots,L$.
\red
\end{secrem}

\subsection{Co-Design Monotonicity and (Differential) Dissipativity}
The next step is to develop a co-design strategy to satisfy the conditions of Theorem \ref{thm:EDD}.
\begin{secprob}\label{prob:EDD}
Given the supply rate $r^\top \delta u - q^\top \delta y$,
find a matrix $K$ with ${\rm rank} \; K = n$ and a state feedback $F$ such that the closed-loop system
\begin{align}
\bar \Sigma_C: 
\left\{\begin{array}{r@{}l}
\dot x &{}= f(x) + B F x + B u,\\
y &{}= C x
\end{array}\right.
\end{align}
is $\cK$-monotone and exponentially differentially dissipative. 
\red
\end{secprob}

\begin{secprop}
Let $\varepsilon_1 > 0$, and 
suppose that $f(x)$ admits the convex relaxation \eqref{eq:cvx}.
Given $(r, q)$, $K$ with ${\rm rank} \; K = n$, and $v > 0$, 
the following LP problem is feasible: for all $i=1,\dots,L$,
\begin{subequations}
\label{opt3:c}
\begin{align}
&\max_{F, P_1, \dots, P_L, H, c_p, c_s} c \nonumber \\
\mbox{s.t. \; }
&K (A_i + B F) = P_i K  \\
&\offd{P_i} \ge c_p \\
&K B \ge c_p \\
&C = H K \\
&H \ge c_p \\
&- q^\top H - v^\top P_i \ge  c_s \\
&r^\top - v^\top K B  \ge c_p\\
&c_p \ge c, \; c_s \ge c + \varepsilon_1. 
\end{align}
\end{subequations}
Then, if $c \ge 0$, the closed-loop system \eqref{sys:cl} is 
$\cK$-monotone and exponentially differentially dissipative.
\end{secprop}
\begin{proof}
The proof is similar to that of Proposition \ref{prop1:ES_alg} and is omitted.
\end{proof}

As in the previous section, for $c\geq 0$, the quantities computed by \eqref{opt3:c} satisfy Theorem \ref{thm:EDD}
(for $P(x) = \sum_{i=1}^L \theta_i(x) P_i$). For $c<0$ we work through a sequence of variational updates based on linearized constraints
to build a suitable cone and feedback pair. The variational approach leads to the following optimization problem:

\begin{subequations} \label{opt4:c}
\begin{align}
&\max_{\delta K, \delta F, \delta v, \delta P_1, \dots, \delta P_L, \delta H, c_p, c_s} c \nonumber \\
\mbox{s.t. \; }
&\delta K (A_i \!+\! B F ) \!+\! K B \delta F =  \delta P_i K\!+\! P_i \delta K \\
&{\rm rank} \; (K + \delta K) = n \\
&\offd{P_i + \delta P_i} \ge c_p \\
&(K + \delta K) B \ge c_p \\
&H K + H \delta K + \delta H K \ge c_p \\
&H + \delta H \ge c_p \\
&v^\top +  \delta v^\top \ge \varepsilon_1\\
&-q^\top \!(H \!+\! \delta H) - (v^\top \!P_i \!+\! \delta v^\top \!P_i \!+\!v^\top\! \delta P_i )\ge  c_s \\
&r^\top - (v^\top K + \delta v^\top K + v^\top \delta K) B  \ge c_p \\
&c_p \ge c, \; c_s \ge c + \varepsilon_1. 
\end{align}
\end{subequations}
where $i = 1,\dots,L$ and with the additional constraint 
that $(\delta K, \delta F, \delta v, \delta P_1, \dots, \delta P_L, \delta H)$ must remain small.

The last constraint and the size of $\delta K$, $\delta F$, $\delta v$, $\delta P_1, \dots, \delta P_L$ and $\delta H$,
and the constraint ${\rm rank} \; K + \delta K =n$ can be 
enforced within an LP setting by following the approach outlined in the previous section (after \eqref{opt2:c}).
The proposed algorithm is summarized below. Remarks \ref{rem:fail_and_complexity} and \ref{rem:OF} apply
also to Algorithm \ref{alg:EDD}.

\begin{algorithm} 
\caption{Solving Problem~\ref{prob:EDD}}
\label{alg:EDD}
\begin{algorithmic}[1]
\small
\Require $\cA := \{ A_1, \dots, A_L\}$, $B$, $(q, r)$, 
initial $K$ with ${\rm rank} \; K =n$, initial $v > 0$, $\varepsilon_1>0$, $\varepsilon_2>0$, $\varepsilon_3>0$
\Ensure Solutions $K$ and $F$ to Problem~\ref{prob:EDD} or {\bf Fail}
\While {{\bf True}}
\State {Solve \eqref{opt3:c}}
\If {$c \ge  0$}
\State {\bf Return $K$ and $F$}
\EndIf
\State {Solve \eqref{opt4:c}}
\State {$K \leftarrow K + \delta K$} and {$v \leftarrow v + \delta v$}
\If {$c$ improves less than $\varepsilon_3$}
\State {\bf Return Fail}
\EndIf
\EndWhile
\end{algorithmic}
\end{algorithm}

\subsection{Example}
We consider again the nonlinear mass spring system \eqref{sys:nmech} of Section~\ref{sec:ex}.
For $q = 1$ and $r = 1$, our objective is to find a simple polyhedral cone $\cK$ and 
a $1$-dimensional dynamic output feedback controller of the form  
\begin{align}
\left\{\begin{array}{r@{}l}
\dot x_c &{}= A_c x_c + B_c y\\
u &{}= C_c x_c + D_c y + \bar u.
\end{array}\right. \quad x_c \in \bR,
\end{align}
such that the closed-loop system is $\cK$-monotone and exponentially differentially dissipative.

Using the aggregate state $x = \begin{bmatrix}x_1 \!&\! x_2 \!&\! x_c\end{bmatrix}^\top$, 
the closed-loop system is well represented by $\bar\Sigma$ in \eqref{eq:nonlinsys}
and its variational system reads
\begin{align*}
\left\{\begin{array}{r@{}l}
\begin{bmatrix}
\dot {\delta x_1} \\ \dot {\delta x_2} \\ \dot {\delta x_c}
\end{bmatrix}
&{}=
\underbrace{\begin{bmatrix}
0 & 1 & 0\\
D_c - \partial k(x_1) & 0 & C_c\\
B_c & 0 & A_c
\end{bmatrix}}_{\partial f(x_1,x_2,x_c)}
\begin{bmatrix}
\delta x_1 \\ \delta x_2 \\ \delta x_c
\end{bmatrix}
+
\begin{bmatrix}
0 \\1 \\ 0
\end{bmatrix}
\delta \bar u \\
\delta y &{}= \delta x_1,
\end{array}\right.
\end{align*}
where $-1 \leq \partial k(x_1) \leq 1$ for all $x_1$. 

We use Algorithm~\ref{alg:EDD} for control design, taking into account Remark~\ref{rem:OF} 
and the range $-1 \leq \partial k(x_1) \leq 1$. 
Thus we iteratively solve \eqref{opt3:c} and \eqref{opt4:c} for the following data:
$i = 1,2$,
\begin{align*}
A_1 + B F &:=
\begin{bmatrix}
0 & 1 & 0\\
D_c - 1 & 0 & C_c\\
B_c & 0 & A_c
\end{bmatrix} \\
A_2 + B F &:=
\begin{bmatrix}
0 & 1 & 0\\
D_c + 1 & 0 & C_c\\
B_c & 0 & A_c
\end{bmatrix} \\
B &:= \begin{bmatrix}
0 & 1 & 0
\end{bmatrix}^\top \\
C &:=
\begin{bmatrix}
1 & 0 & 0
\end{bmatrix},
\end{align*}
and a suitable selection of small $\varepsilon_1$, $\varepsilon_2$, and $\varepsilon_3$.

The initial $K$ is
\begin{align*}
K 
:= 
\begin{bmatrix}
C \\ \bar K
\end{bmatrix}, \;
\bar K :=
\begin{bmatrix}
0 &  0.707 & 0.707\\
0 & -0.707 & 0.707\\
-0.416 &  0.572 &  0.707\\
-0.416 & -0.572 &  0.707\\
-0.672 &  0.219 &  0.707\\
-0.672 & -0.219 &  0.707
\end{bmatrix}.
\end{align*}
Since the first row of $K$ is $C$, the constraints $C = H K$ and $H \ge 0$ hold for
\begin{align*}
H := 
\begin{bmatrix}
1 & 0 & \cdots & 0 
\end{bmatrix}.
\end{align*}
Thus, if we update only $\bar K$, then $C = H K$ and $H \ge 0$ always hold,
and the constraints for $H$ and $\delta H$ can be removed from
\eqref{opt3:c} and \eqref{opt4:c}.

\begin{figure}[tb]
\begin{center}
\includegraphics[width=0.68\columnwidth]{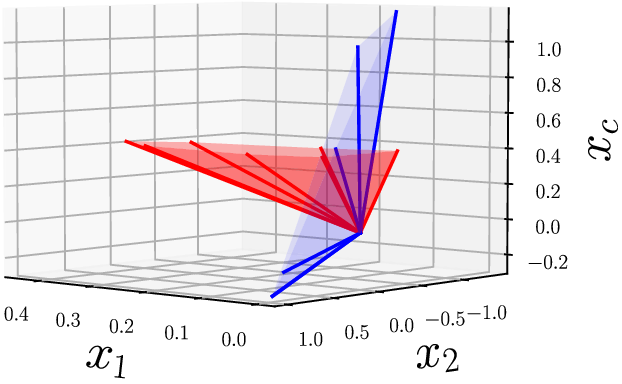}\vspace*{2mm}
\caption{Initial cone in red and computed cone in blue.}
\label{fig:cone}
\end{center}
\end{figure}

We apply a modification of Algorithm~\ref{alg:EDD}.
A solution to Problem \ref{prob:EDD} is obtained as
\begin{align*}
\bar K 
&=
\begin{bmatrix}
0.995 & 0.102 & -0.00331\\
0.989 & 0.140&  0.0542\\
0.976 & 0.176 & 0.129\\
0.935 &  0.281 &  0.216\\
-0.563 & 0.230 &  0.794\\
-0.531&  0.234 &  0.814
\end{bmatrix}\\
v
&=
[\begin{matrix}
0.426 & 0.104 & 0.395
\end{matrix}\\
&\qquad \begin{matrix}
0.0982 & 0.230 & 0.870 & 2.59
\end{matrix}]^\top\\
\begin{bmatrix}
A_c & B_c\\
C_c  & D_c
\end{bmatrix}
&=
\begin{bmatrix}
-30.0 &  26.1\\
10.7 & -10.0
\end{bmatrix}.
\end{align*}
The dynamic output feedback controller above enforces $\cK$-monotonicity and exponential differential dissipativity
of the closed-loop system.
The initial and obtained cones are plotted in Fig.~\ref{fig:cone}.

\begin{secrem}
The cone derived by Algorithm~\ref{alg:EDD} depends on the initial cone, since Problem \ref{prob:EDD} has multiple solutions. 
We have tested our algorithm for different initial cones.
We have tested our algorithm from $p_\cK = 3$ by increasing the number of $p_\cK$. Until $p_\cK = 7$, we have not succeeded to obtained a solution.
For $p_\cK = 7$, each time the algorithm converges to a suitable cone/controller pair. 
These tests have shown that the elements of $F$ could become large during the iteration, which causes overflow. 
This problem can be handled by enforcing lower and upper bounds on its each element
(we have used $\pm 30$ bounds for the case above).
\red
\end{secrem}

\color{black}
\section{Conclusion}
We have revised and extended methods for analysis and synthesis of monotone systems.
We have shown that monotone systems can be variationally embedded into
positive systems. Then, using this embedding, 
we have derived linear conditions for monotonicity with respect to generic cones $\cK$,
removing the usual restriction to the positive orthant $\bR^n_+$ (cooperative systems).
We have also derived linear, i.e. scalable, conditions for incremental stability 
and differential dissipativity of general monotone systems. 
Our results allow for analysis and synthesis
of networks of monotone systems. From the control design perspective, this means that
we can stabilize a network via control design at the level of a single node/component. 
The theoretical results of the paper have been illustrated by three examples, each
capturing different scenarios. 

Theory of monotone systems is well established but often limited by the challenge of 
recognizing when a system is monotone. Simple conditions exist for monotonicity with 
respect to partial orders induced by an orthant \cite{AS:03}. In this paper we also show
that, given a simple polyhedral cone $\cK$, establishing $\cK$-monotonicity 
reduces to a linear program. However, the problem of demonstrating the \emph{existence} of a
cone $\cK$ that guarantee $\cK$-monotonicity of a nonlinear system is way harder. 
In this paper we present algorithms to answer such a question, taking advantage
of a co-design procedure for cone and feedback controller. 
The algorithms are sound but by no means provide a complete answer to the question.
They require further research to overcome limitations related to local minima and to issues of numerical stability. 
For developing numerical algorithms, we have focused on constant cones and linear feedback controllers. The theory can be extended to differential cones (i.e., non-constant cones $\cK(x)$) \cite{FS:16} and nonlinear feedback control design, which will be reported in future publication.

\appendix
\subsection{Proof of Theorem \ref{thm:NM}}\label{proof:NM}
We use Theorem~\ref{thm:mbp}.

\emph{(if part)} 
From \eqref{condf:NM},
the embedding $\delta z = K \delta x$ yields
\begin{align*}
d\bar \Sigma^P: 
\left\{\begin{array}{r@{}l}
\dot {\delta z} &{}= P(x(t)) \delta z + G \delta u\\
\delta y &{}= H \delta z,
\end{array}\right.
\end{align*}
where $\offd{P(\cdot)} \ge 0$ and $G, H \ge 0$.
Therefore, it follows that
\begin{align*}
\begin{array}{c}
K \delta x(0) \in \bR_+^{p_{\cK}}, \; \delta u(t)  \in  \bR_+^m, \; \forall t \! \ge \! 0\\
\implies \\
K \delta x(t) \in  \bR_+^{p_{\cK}}, \; \delta y(t) \in \bR_+^p, \; \forall t \! \ge \! 0.
\end{array}
\end{align*} 
Since $K \delta x \in \bR_+^{p_{\cK}}$ implies $\delta x \in \cK$, this is nothing but $\cK$-monotonicity.

\emph{(only if part)} 
Let $\delta u(t) = 0$, $t \ge 0$.
We first show that for each $x \in \bR^n$ and every $i=1,\dots, p_\cK$, there exists $c_i(x) > 0$ such that
\begin{align}\label{pf1:mbp}
(K_i)^\top (\partial f(x) + c_i(x) I_n) \delta x \ge 0 \quad \nonumber\\
\forall \delta x \in \bR^n
\mbox{ s.t. }
K \delta x \ge 0.
\end{align}
Recalling $\cK = \{ \delta x \in \bR^n : K \delta x \ge 0\}$, we consider two cases (i) $(K_i)^\top \delta x > 0$, $\delta x \in \cK$ and (ii) $(K_i)^\top \delta x = 0$, $\delta x \in \cK$.
In case (i), there exists a sufficiently large $c_i(x) > 0$ such that
\begin{align*}
(K_i)^\top (\partial f(x) + c_i(x) I_n) \delta x \ge 0  \quad \nonumber\\
\forall \delta x \in \cK
\mbox{ s.t. }
(K_i)^\top \delta x > 0.
\end{align*}
In case (ii), $\cK$-monotonicity implies
\begin{align*}
\begin{array}{c}
(K_i)^\top \partial f(x) \delta x \ge 0, \quad
\forall \delta x \in \cK
\mbox{ s.t. }
(K_i)^\top \delta x = 0.
\end{array}
\end{align*}
Noting $(K_i)^\top \delta x = 0$, we obtain
\begin{align*}
(K_i)^\top ( \partial f(x) \delta x + c_i(x) I_n) \delta x \ge 0  \quad \nonumber\\
\forall \delta x \in \cK
\mbox{ s.t. }
(K_i)^\top \delta x = 0.
\end{align*}
Thus, we have~\eqref{pf1:mbp}.

Next, from~\eqref{pf1:mbp}, Farkas's lemma \cite[Corollary 7.ld]{Schrijver:98} implies that 
for each $i=1,\dots, p_\cK$, there exists $P_i: \bR^n \to \bR^{p_{\cK}}$ such that
\begin{align*}
(K_i)^\top (\partial f(x) + c_i(x) I_n) = P_i^\top (x) K,
\quad
P_i^\top (x) \ge 0,
\end{align*}
i.e.,
\begin{align*}
&\underbrace{\begin{bmatrix}
(K_1)^\top \\
\vdots\\
(K_{p_\cK})^\top
\end{bmatrix}}_{K}\partial f(x)
+
\begin{bmatrix}
c_1(x) & \\
& \ddots \\
& & c_{p_\cK}(x)
\end{bmatrix}
\underbrace{\begin{bmatrix}
(K_1)^\top \\
\vdots\\
(K_{p_\cK})^\top
\end{bmatrix}}_{K}\\
&=
\begin{bmatrix}
(P_1)^\top (x)\\
\vdots\\
(P_{p_\cK})^\top  (x)
\end{bmatrix}K.
\end{align*}
Therefore, \eqref{condf:NM} holds for
\begin{align*}
P(x) &:= \begin{bmatrix} P_1(x) & \dots & P_{p_{\cK}} (x) \end{bmatrix}^\top \\
&\qquad - {\rm diag}\{c_1(x),\dots,  c_{p_\cK} (x)\}.
\end{align*}

Next, $\cK$-monotonicity for $\delta x(0)=0$ implies \eqref{condB:LM}. 
Finally, $\delta y \in \bR^p$ for any $\delta x \in \cK$ implies \eqref{condC:LM}. 
\hfill \QED

\subsection{Proof of Theorem~\ref{thm:IES}}\label{proof:IES}
First, we consider the embedded variational system $\dot{\delta z} = P (x(t)) \delta z$,
which is a positive linear time-varying system for every trajectory $x(\cdot)$.

Define the function $V : \bR_+^{p_{\cK}} \to \bR_+$ given by $V(\delta z) = v^\top \delta z$ and note that  
$V(\delta z)> 0$ for all $\delta z\in \bR_+^{p_{\cK}} \setminus\{0\}$.
Consider any trajectory $\delta z(\cdot)$ such that $\delta z(0) \in \bR_+^{p_{\cK}}$. 
Positivity guarantees that $\delta z(t) \in \bR_+^{p_{\cK}}$ for all $t \geq 0$,
thus \eqref{cond:IES} guarantees that 
$\dot{V} = v^\top P(x(t))\delta z(t) \leq -\lambda \delta z(t)$ for all $t \geq 0$.
It follows that there exists $k' \ge 1$ such that $| \delta z(t) |_1 \le k' e^{- \lambda t} | \delta z(0)|_1$ for all $\delta z(0) \in  \bR_+^{p_{\cK}}$.

Take now a generic trajectory $\delta z(\cdot) \in  \bR^{p_{\cK}}$ and note that it can be decomposed into
$\delta z(\cdot) = \delta z_1(\cdot) - \delta z_2(\cdot)$ where $\delta z_i(0) \in \bR_+^{p_{\cK}}$, $i \in \{1,2\}$.
Thus, there exists $k' \ge 1$, such that 
$
| \delta z(t) |_1 \le k' e^{- \lambda t} | \delta z(0)|_1.
$
See also \cite[Corollary 4.8]{KBC:20}.

Substituting $\delta z(0) = K \delta x(0)$ into the above yields
\begin{align*}
| K \delta x(t) |_1 \le k' e^{- \lambda t} | K \delta x(0)|_1.
\end{align*}
Since ${\rm rank}\; K =n$, there exists $k \ge 1$ such that
\begin{align}\label{pf1:IES}
| \delta x(t) |_1 \le k e^{- \lambda t} |  \delta x(0)|_1
\end{align}
for any $t \ge 0$ and $\delta x(0)  \in \bR^{p_{\cK}}$ as long as $x(t)$ exists.

Next, we show that \eqref{pf1:IES} implies the existence of $x(t)$. 
Since $\delta x = f(x(t))$ satisfies $\dot{\delta x} = \partial f (x(t)) \delta x$, it follows from \eqref{pf1:IES} that
\begin{align*}
| f(x(t)) |_1 \le k e^{- \lambda t} | f(x(0)) |_1,
\end{align*}
which implies that $\dot x(t) = f(x(t))$ is a bounded function of $t \ge 0$. 
Thus, $x(t)$ exists for any $t \ge 0$ and $x(0) \in \bR^n$.

Finally, we consider the line segment $\gamma (s) = s x + (1 - s) x'$, $s \in [0,1]$.
Then, the solution to $\dot x = f(x)$ starting from $\gamma (s)$, denoted by $\phi(t, \gamma (s))$, satisfies
\begin{align*}
\frac{\partial}{\partial t} \frac{\partial \phi(t, \gamma (s))}{\partial s}
= \frac{\partial f(\phi(t, \gamma (s)))}{\partial \phi} \frac{\partial \phi(t, \gamma (s))}{\partial s}.
\end{align*}
Substituting $\delta x = \partial \phi/\partial s$ into \eqref{pf1:IES} and 
taking the integration with respect to $s$ over $[0, 1]$ lead to
\begin{align*}
| x(t) - x'(t) |_1 &= \left| \int_0^1 \frac{\partial \phi(t, \gamma (s))}{\partial s}ds \right|_1  \\
&\le \int_0^1 \left| \frac{\partial \phi(t, \gamma (s))}{\partial s} \right|_1 ds\\
& \le k e^{- \lambda t} \int_0^1 \left| \frac{\partial \gamma (s)}{\partial s} \right|_1 ds \\
& \le k e^{- \lambda t} | x(0) - x'(0) |_1.
\end{align*}
for any $t \ge 0$ and $(x(0), x'(0) ) \in \bR^n \times \bR^n$.
\hfill \QED

\subsection{Proof of Corollary~\ref{cor:IES}}\label{proof:IES_cor}
We consider the embedded variational system  $\dot{\delta z} = P (x(t)) \delta z$ by $\delta z = K \delta x$ and also $\dot{\delta \bar z} = |\offd{P (x(t))}| \delta \bar z$. For the same initial condition $\delta z(0) = \delta \bar z(0) = \delta z_0 \in \bR^n$, it follows that $- \delta \bar z(\cdot) \le \delta z(\cdot) \le \delta \bar z (\cdot)$. Moreover, according to the proof of Theorem~\ref{thm:IES}, \eqref{cond:IES2} implies 
for some $k' \ge 1$, $| \delta \bar z(t) |_1 \le k' e^{- \lambda t} | \delta z_0|_1$ for all $\delta z_0\in  \bR^{p_{\cK}}$.
Therefore, we have $| \delta z(t) |_1 \le k' e^{- \lambda t} | \delta z_0|_1$. The rest is the same as the proof of Theorem~\ref{thm:IES}.
\hfill \QED

\subsection{Proof of Theorem~\ref{thm:EDD}}\label{proof:EDD}
Taking into account \eqref{condf:NM}, 
\eqref{cond2:EDD} and \eqref{cond3:EDD} imply that
\begin{align*}
v^\top P(x) K \delta x + \lambda v^\top K \delta x &\le - q^\top H K \delta x\\
v^\top G \delta u &\le  r^\top \delta u
\end{align*}
for all $(\delta x, \delta u) \in  (\cK,  \bR_+^m)$.
It follows that
\begin{align*}
v^\top ( P (x)  K \delta x +  G \delta u) + \lambda v^\top K \delta x 
&\le r^\top \delta u - q^\top H K \delta x.
\end{align*}
That is,
\begin{align*}
v^\top K ( \partial f (x) \delta x +  B \delta u) + \lambda v^\top K \delta x 
&\le r^\top \delta u - q^\top C \delta x.
\end{align*}
which corresponds to the dissipation inequality \eqref{eq2:EDD}.
\eqref{eq:EDD} follows by integration.
\hfill \QED

\subsection{Proof of Theorem \ref{thm3:net}}\label{proof3:net}

First, we show monotonicity using Theorem \ref{thm:NM}.
The $i$th subsystem of the interconnected system is
\begin{align*}
\dot x_i = f_i (x_i ) + B_i W_i ( y_1, \dots, y_N),
\end{align*}
and its variational system is
\begin{align*}
\dot {\delta x}_i &= \partial f_i (x_i (t) ) \delta x_i \\
&\quad
+ B_i \sum_{j = 1}^N \partial_{y_j} W_i(y_1(t), \dots, y_N(t))  C_j \delta x_j.
\end{align*}
From \eqref{condf:NM}, the embedding $\delta z_i = K_i \delta x_i$ yields
\begin{align*}
\dot {\delta z}_i &= P_i ( x_i (t) ) \delta z_i \\
&\quad
+ G_i \sum_{j = 1}^N \partial_{y_j} W_i(y_1(t), \dots, y_N(t))  H_j \delta z_j.
\end{align*}
Its compact form is represented by $\dot {\delta z} = P(x (t)) \delta z$, where
\begin{align}\label{pf3:net1}
P(x) &:= 
{\rm diag} \{P_1(x), \dots, P_N(x)\} \nonumber\\
&\qquad
+
\begin{bmatrix}
G_1 \partial_{y_1} W_1(\cdot) H_1 &  \ddots & G_1 \partial_{y_N} W_1(\cdot) H_N\\ 
\vdots & \ddots & \vdots\\
G_N \partial_{y_1} W_N(\cdot)  H_1 & \cdots & G_N \partial_{y_N} W_N (\cdot) H_N 
\end{bmatrix}.
\end{align}
$\cK_i$-monotonicity, $i=1,\dots,N$ (with Theorem \ref{thm:NM}) and item (ii) imply $\offd{P(\cdot )} \ge 0$, and thus 
the interconnected system is $(\cK_1 \times \cdots \times \cK_N)$-monotone.
Note that the pair of items (iia) and (iib) is equivalent to $\offd{P(\cdot )} \ge 0$.

Next, we show incremental exponential stability.
Define the collective vector $v := [\begin{matrix}v_1^\top & \cdots & v_N^\top \end{matrix}]^\top > 0$.
For the $i$th component,
it follows from $G_i, H_i \ge 0$ and items i) and iii) that 
\begin{align}\label{pf3:net2}
&v_i^\top P_i + \sum_{j=1}^N v_j^\top G_j \partial_{y_i} W_j H_i \nonumber\\
&\le v_i^\top P_i + \sum_{j=1}^N r_j^\top \partial_{y_i} W_j H_i \nonumber\\
&\le v_i^\top P_i + q_i^\top H_i \le - \lambda_i v_i^\top.
\end{align}
In other words, we have
\begin{align*}
v &> 0\\
- v^\top P(x) &> \min_{i=1,\dots,N} \{ \lambda_i \} v^\top,
\quad 
\forall x \in \bR^{n_1 + \cdots + n_N}.
\end{align*}
Therefore, from Theorem~\ref{thm:IES}, the networked interconnection is incremental exponential stable.

Finally, we consider the case where item (ii) is replaced by items (iia) and (iib).
Item (i) with $r_i^\top = v_i^\top G_i$, $H_i \ge 0$, and item (iii) for $i=1,\dots,N$ imply~\eqref{pf3:net2}.
\hfill \QED

\subsection{Proof of Theorem \ref{thm4:net}}\label{proof4:net}
As shown in Appendix~\ref{proof3:net}, the embedded networked interconnection 
$\dot {\delta z} = P(x (t)) \delta z$ with $P(x)$ in \eqref{pf3:net1} satisfies $\offd{P(\cdot )} \ge 0$ if item (ii) holds.
Next, we show that 
if items (i) and (iii) hold, then
the vector $w := [\begin{matrix}w_1^\top & \cdots & w_N^\top \end{matrix}]^\top $
satisfies
\begin{subequations}\label{cond:IESw}
\begin{align}
w &> 0, \label{cond1:IESw}\\
- P(x) w &\ge \lambda w,  \label{cond2:IESw}
\quad 
\forall x \in \bR^n, \; n:=n_1 + \cdots + n_N.
\end{align}
\end{subequations}
\eqref{cond1:IESw} follows from $w_i > 0$, $i=1,\dots, N$ in item (i).
We verify the second inequality for the $i$th block component,
it follows from $G_i, H_i \ge 0$ and items i) and iii) that 
\begin{align*}
&P_i w_i + G_i \sum_{j=1}^N \partial_{y_j} W_i  H_j w_j \nonumber\\
&\le P_i w_i + G_i \sum_{j=1}^N \partial_{y_j} W_i  r_j\nonumber\\
&\le P_i  w_i + G_i q_i \le - \lambda_i w_i.
\end{align*}
Therefore, we have \eqref{cond2:IESw} for $\lambda = \min_{i = 1, \dots N} \{ \lambda_i \}$.
In the case where item (ii) is replaced by items (iia) and (iib), 
\eqref{cond2:IESw} can be shown from item (i) with $r_i = H_i w_i$, $G_i \ge 0$, and item (iii) for $i=1,\dots,N$.

Finally, we show that \eqref{cond:IESw} implies the incremental exponential stability of the networked interconnection. 
Define the function $V : \bR_+^n \to \bR_+$ given by $V(\delta z) = \max_{i=1,\dots,n}\{ \delta z_i/w_i\}$ and note that  
$V(\delta z)> 0$ for all $\delta z\in \bR_+^n\setminus\{0\}$.
For this function, we define the following set of the indexes, $J(\delta z) =\{ j=1, \dots, n : \delta z_j/w_j = V(\delta z)  \}$.
From its definition, we have
\begin{align*}
\frac{\delta z_i}{w_i}
\le
\frac{\delta z_j}{w_j}
\quad
\iff
\quad
\delta z_i \le \frac{w_i}{w_j} \delta z_j
\end{align*}
for all $i=1,\dots, n$ and $j \in J(\delta z)$.
Then, the upper right Dini derivative of $V(\delta z)$ (see. e.g., \cite{Danskin:66}) satisfies
\begin{align*}
D^+ V(\delta z) 
= \max_{j \in J(\delta z)}  \frac{\delta \dot z_j}{w_j}
&= \max_{j \in J(\delta z)}  \frac{1}{w_j} \sum_{i=1}^n P_{j,i}( \cdot ) \delta z_i\\
&\le \max_{j \in J(\delta z)}  \frac{1}{w_j^2} \sum_{i=1}^n P_{j,i}( \cdot ) w_i \delta z_j,
\end{align*}
and thus, from \eqref{cond2:IESw}, 
\begin{align*}
D^+ V(\delta z) 
\le - \lambda \max_{j \in J(\delta z)}  \frac{\delta z_j}{w_j}
= - \lambda V(\delta z).
\end{align*}
This and positivity guarantee that there exists $k'\ge 1$ such that $| \delta z(t) |_{\infty} \le k' e^{- \lambda t} | \delta z(0)|_{\infty}$ for all $\delta z(0) \in  \bR_+^n$.
From the equivalence between the vector $1$- and $\infty$-norms, the inequality also holds with respect to $| \cdot |_1$.
Thus, incremental exponential stability follows as in the proof of Theorem~\ref{thm:IES}.
\hfill \QED

\subsection{Proof of Proposition \ref{prop1:ES_alg}}\label{proof:ES_alg}
First, we show feasibility.
For any $K$ with ${\rm rank} \; K = n$, $A_i$, $B$, and $F$, there exists $P_i$ satisfying the first constraint.
For any $P_i$, $v$ and $\varepsilon_1 > 0$, there exist $c$, $c_p$, and $c_s$ satisfying the other constraints. 

Next, if $c \ge 0$, then $c_p \ge 0$ and $c_s > 0$.
From \eqref{eq:cvx}, for each $x \in \bR^n$, there exist $\theta_i(x)$, $i=1,\dots,L$ such that
$\partial f(x) = \sum_{i=1}^L \theta_i(x) A_i$.
Therefore, the first and second constraints with $c_p \ge 0$ imply \eqref{condA:NM} for the closed-loop system \eqref{sys:cl}, 
where $P(x) = \sum_{i=1}^L \theta_i(x) P_i$.
The third one with $c_s > 0$ implies \eqref{cond:IES} for some $\lambda > 0$.
\hfill \QED

\bibliographystyle{IEEEtran}
\bibliography{PositiveStability}

\end{document}